\documentclass[uslettersize, 10 pt, conference]{ieeeconf}

\IEEEoverridecommandlockouts                              
\usepackage{xcolor}
\usepackage{graphics} 
\usepackage{epsfig} 
\usepackage{amsmath} 

\usepackage{amssymb,amsthm}  
\usepackage{amscd}
\usepackage[noadjust]{cite}
\usepackage{float}
\usepackage{times}
\usepackage{nicefrac}
\newtheorem{theorem}{Theorem}
\newtheorem{remark}{Remark}
\newtheorem{proposition}{Proposition}
\newtheorem{lemma}{Lemma}
\newtheorem{corollary}{Corollary}

\newcommand{\R}{\mathbb{R}}

\newcommand{\tr}{\operatorname{tr}}
\newcommand{\bb}{\mathbb}

\newcommand{\diag}{\operatorname{diag}}
\newcommand{\sgn}{\operatorname{sgn}}
\newcommand{\SO}{{\rm SO}}
\newcommand{\so}{\mathfrak{so}}
\newcommand{\ad}{\operatorname{ad}}
\newcommand{\oc}{\Omega^C}
\newcommand{\ob}{\Omega^B}

\renewcommand{\cal}{\mathcal}

\title{\LARGE \bf Joint Actuator-sensor Design for Stochastic Linear Systems
}

\author{Xudong Chen$^*$ 
\thanks{$^*$X. Chen is with the Department of ECEE, University of Colorado, Boulder. Email: xudong.chen@colorado.edu. 
}}

\begin{document}

\maketitle
\thispagestyle{empty}
\pagestyle{empty}

\begin{abstract}
We investigate the joint actuator-sensor design problem for stochastic linear control systems. Specifically, we address the problem of identifying  a pair of sensor and actuator which gives rise to the minimum expected value of a  quadratic cost.   It is well known that for the linear-quadratic-Gaussian (LQG) control problem, the optimal feedback control law can be obtained via the celebrated separation principle. Moreover, if the system is stabilizable and detectable, then the infinite-horizon time-averaged cost exists. But such a cost depends on the placements of the sensor and the actuator. We formulate in the paper the optimization problem about minimizing the time-averaged cost over admissible pairs of actuator and sensor under the constraint that their Euclidean norms are fixed. The problem is non-convex and is in general difficult to solve. We obtain in the paper a gradient descent algorithm (over the set of admissible pairs) which minimizes the time-averaged cost. Moreover, we show that the algorithm can lead to a unique local (and hence global) minimum point under certain special conditions. 
	\end{abstract}


\section{Introduction}
The problem of system design is, roughly speaking, to optimize its intrinsic parameters, such as placement of actuator and/or sensor, so as to either minimize a certain cost function (e.g., energy consumption) or to maximize a certain performance measure (e.g., sensing accuracy). When the system is networked, comprised of several physical entities (such as a swarm of robots or unmanned aerial vehicles), allocation of communication resource can be also considered as an intrinsic parameter. Optimal resource allocation/communication scheduling has also been addressed widely in the literature (see, for example,~\cite{brockett1995stabilization, zhang2006communication, he2006sensor, shi2013optimal, han2017optimal, chen2017optimal}).  

We focus in the paper a joint actuator-sensor design problem for the following stochastic linear system over admissible {\em actuator vector} $b$ and {\em sensor vector} $c$:
\begin{equation*}
\left\{
\begin{array}{l}
dx_t = Ax_t dt + bu(t) dt + dw_t,\\
dy_t = c^{\top} x_t dt + d\nu_t.
\end{array}
\right. 
\end{equation*}
We aim to minimize an infinite horizon time-averaged quadratic cost function: 
$$
\eta_\infty := \lim_{T\to \infty}\bb{E}\left(\frac{1}{T}\int^T_0 (x^\top x + u^2) dt\right) 
$$ 
For fixed $b$ and $c$, this is known as  the {\em linear-quadratic-Gaussian (LQG)} control problem. The optimal feedback control law can be obtained via the separating principle. However, our goal here is {\em not} to reproduce the analysis for deriving such optimal feedback control law. But rather, we assume that such control law has been employed, and we address the problem of how to minimize the cost over the pairs $(b,c)$ under the constraint that $|b|$ and $|c|$, i.e., the Euclidean norms of $b$ and of $c$, are fixed. A precise formulation of the joint actuator-sensor design problem will be given shortly. 

{\em Literature review.} We note here that similar problems of actuator-design or sensor-design (but not jointly) have also been addressed recently. We first refer the reader to~\cite{belabbas2016geometric} for the optimal sensor design problem. A gradient descent algorithm was derived there, which is proven to possess a unique exponentially stable equilibrium for the case where $A$ is Hurwitz and $|c|$ is relatively small.     We also refer the reader to~\cite{chen2017optimalact} for the problem of optimizing the actuator vector $b$ which requires minimal energy to drive the system from an initial condition in the unit sphere to the original in the worst case (with respect to the choice of the initial condition). A complete solution was provided for the case where $A$ is positive definite with distinct eigenvalues.   We further refer to~\cite{chen2014fluid,hiramoto2000optimal,rao1991optimal} for actuator/sensor design problems which are application-specific.     

Amongst other related problems, we mention the actuator and/or sensor selection problem. The problem there is to select a small number of actuators/sensors out of a large discrete set so as to minimize the control energy or to maximize the sensing accuracy.   For example, 
the authors in~\cite{pasqualetti2014controllability} established  lower bounds for control energy for a given selection of actuators. Similar problem, but with the focus on sensing accuracy, was addressed in~\cite{tzoumas2016sensor}.  A key difference between the actuator/sensor design problem and the selection problem is that the solution space of the former is usually a non-convex continuous space while the latter is in general a combinatorial optimization problem. Thus, the techniques and mathematical tools used  in these two classes of problems are quite different. We further note that greedy type of algorithms were widely used in sensor/actuator selection problems. For example, we refer the reader to~\cite{olshevsky2014minimal} for the minimal controllability problem (i.e., the problem of selecting minimal number of variables so that the resulting linear system is controllable), and to~\cite{zhang2017sensor} for the sensor selection problem for Kalman filtering.

{\em Outline of contribution and organization of the paper.} The contribution of the work is the following: First, we formulate the joint actuator-sensor problem in Section~II. In particular, we provide an explicit expression of the cost function and identify the solution space as a coadjoint orbit equipped with the so-called normal metric. We then derive in Section~III the gradient flow over the solution space with respect to the given metric. We also provide analytical results about the gradient flow. In particular, we characterize conditions for a point in the solution space to be an equilibrium of the gradient flow. To illustrate the type of analysis one needs to carry out, we focus on a special class of linear dynamics where the system matrix $A$ is negative definite and $|b|, |c|$ are relatively small. We show that in such case, there is a unique stable equilibrium of the associated gradient flow. In particular, the optimal actuator vector and sensor vector are aligned with the eigenvector of the matrix $A$ with respect to its largest eigenvalue. These results, as well as the analysis, are given in Section IV. We provide conclusions at the end.

\section{Preliminaries and problem formulation}

We formulate here the joint actuator-sensor design problem. 
To start, we first have a few preliminaries about the  classic linear-quadratic-Gaussian (LQG) control problem.
\subsection{Preliminaries about LQG control}

Consider a continuous-time linear stochastic system with a continuous-time measurement output: 
\begin{equation}\label{eq:controlsystem}
\left\{
\begin{array}{l}
dx_t = Ax_t dt + bu(t) dt + dw_t\\
dy_t = c^{\top} x_t dt + d\nu_t
\end{array}
\right. 
\end{equation}
where $x_t\in \R^n$ is the state, $u(t)\in \R$ is the control input, $y_t\in \R$ is the measurement output, and $w_t\in \R^n$, $\nu_t\in \R$ are independent standard Wiener processes. 
We call the vectors $b, c\in \R^n$ the {\em actuator and sensor vectors}, respectively.  Next, consider the expected value of a quadratic cost function:    
$$
\eta_T := \bb{E}\left(\frac{1}{T}\int^T_0 (x^\top x + u^2) dt\right) 
$$

The so-called {\em LQG control problem} is about finding an optimal feedback control law $u^*(t)$ which minimizes the above cost. It is well known that 
the optimal control problem can be solved via the celebrated separation principle: Let $K(t)$ and $\Sigma(t)$ be two {\em differential Riccati equations} defined as follows:
\begin{equation}\label{eq:differentialequations}
\left\{
\begin{array}{l}
\dot K(t) = -A^{\top} K(t) - K(t) A + K(t)bb^{\top}K(t) - I,   \\
\dot \Sigma(t) = A \Sigma(t) + \Sigma(t) A^{\top} - \Sigma(t) cc^{\top} \Sigma(t) + I, 
\end{array}
\right. 
\end{equation}
where the boundary conditions are specified by $K(T) = 0$ and $\Sigma(0)$ is the covariance matrix of $x_0$. Then, an optimal feedback control law $u^*(t)$ is given by 
$$
u^*(t):= -b^\top K(t) \hat x_t,
$$ 
where $\hat x_t$ is the {\em minimum  mean-squared-error estimate} of $x_t$, which is given by (see~\cite{kalman1961new}) 
\begin{multline*}
d\hat x_t = A \hat x_t dt + b u^*(t) dt + \Sigma(t) c (dy_t - c^\top x_t dt) \\
 = (A  - b^\top K(t) ) \hat x_t  dt + \Sigma(t) c (dy_t - c^\top x_t dt).
\end{multline*}
Moreover, under such an optimal control, the (minimized) cost function is given by  
\begin{equation}\label{eq:costfunction}
\eta^*_T := \frac{1}{T}\int^T_{0}\tr( K(t) bb^\top K(t)\Sigma(t) + K(t))dt.
\end{equation}
where $\tr(\cdot)$ denotes the trace of a matrix. 

Further, if the control system~\eqref{eq:controlsystem} is {\em stabilizable}  and {\em detectable},  then the steady-states of the differential Riccati equations~\eqref{eq:differentialequations} exist, which are the unique positive semi-definite (PSD) solutions to the following {\em algebraic Riccati equations (AREs)}:
\begin{equation}\label{eq:algebraicriccati}
\left\{
\begin{array}{l}
 A^{\top} K + K A - Kbb^{\top}K + I = 0,   \\
A \Sigma + \Sigma A^{\top} - \Sigma cc^{\top} \Sigma + I = 0. 
\end{array}
\right. 
\end{equation} 
It follows that the limit of $\eta^*_T$ also exists, which we state in the following Lemma:

\begin{lemma}
If system~\eqref{eq:controlsystem} is stabilizable and detectable, then 
\begin{equation}\label{eq:performance}
\Phi:= \lim_{T\to\infty} \eta^*_T = \tr(A^{\top}K\Sigma + \Sigma KA + K + \Sigma) 
\end{equation}
where $K$ and $\Sigma$ are the PSD solutions to~\eqref{eq:algebraicriccati}. 
\end{lemma}

\subsection{Problem formulation: joint actuator-sensor design.} Note that the value of $\Phi$ defined in~\eqref{eq:performance} depends on the actuator and sensor vectors $b$ and $c$, via the solutions $K$ and $\Sigma$ to the AREs~\eqref{eq:algebraicriccati}. We will thus write $K(b)$, $\Sigma(c)$, and $\Phi(b,c)$ on occasions to indicate  such dependence explicitly. The optimal joint actuator-sensor design problem we address in the paper is an optimization problem about minimizing the function $\Phi(b,c)$ over all admissible pairs  $(b,c)$. 

To proceed, we first note the fact that $\Phi(b,c)$ decreases if we increase the norms of $b$ and $c$. Specifically, we fix a pair of actuator and sensor vectors $(b,c)$, with $(A,b)$ stabilizable and $(A,c)$ detectable. With slight abuse of notation, we denote by $K(r)$ and $\Sigma(s)$, for $r, s> 0$,  the PSD solutions to the following AREs:
\begin{equation}\label{eq:parametrize}
\left\{
\begin{array}{l}
 A^{\top} K(r) + K(r) A - r K(r)bb^{\top}K(r) - I = 0,   \\
A \Sigma(s) + \Sigma(s) A^{\top} - s\Sigma(s) cc^{\top} \Sigma (s)+ I = 0. 
\end{array}
\right. 
\end{equation}
 Then, we have $$dK(r)/dr \le 0 \quad \mbox{ and } \quad d\Sigma(s)/ ds \le 0.$$ 
 We refer to~\cite{wredenhagen1993curvature} or Prop. 3 of~\cite{chen2017optimal} for a prove of the above inequalities. We also gave in~\cite{chen2017optimal} generic conditions on when the inequalities are strict. If we let $\Phi(r, s)$ be defined as in~\eqref{eq:performance}, with $K$ and $\Sigma$ replace by $K(r)$ and $\Sigma(s)$, then we have the following fact:
 
 \begin{lemma}\label{lem:fixnorm}
 For fixed $b$ and $c$ with $(A,b)$ stabilizable and $(A,c)$ detectable, we have 
 $$\partial \Phi(r,s)/\partial r \le 0, \qquad \partial \Phi(r,s)/\partial s \le 0.$$
 The inequalities are strict if $K'(r) < 0$ and $\Sigma'(s) < 0$.  
 \end{lemma}
 
 \begin{proof}
We focus only on the proof for $\partial \Phi(r, s) /\partial r \le 0$. By symmetry, the same argument can be applied to establish $\partial \Phi(r, s) /\partial s \le 0$.      
 For convenience, we let $K'(r) := d K(r)/dr$. We obtain by computation
 \begin{multline*}
 \partial \Phi(r,s)/\partial r  = \tr((A\Sigma(s) + \Sigma(s)A^\top + I)K'(r)) \\
 = s\tr(\Sigma(s) cc^\top \Sigma(s) K'(r)) \le 0
 \end{multline*}
 where the second equality comes from~\eqref{eq:parametrize}, and the last inequality comes from the fact that $\tr(PQ) \le 0$ for $P\ge 0$ and $Q \le 0$. Here, $P := \Sigma(s) cc^\top \Sigma(s)$ and $Q := K'(r)$. The inequalities are strict if $Q <0 $ and $P \neq 0$. Note that $P \neq 0$ if and only if $\Sigma(s) c \neq 0$.  Also, note that by computation (see, for example,~\cite{chen2017optimal})
 $$
 \Sigma'(s) = -\int^\infty_0 e^{(A -s \Sigma cc^\top)t} \Sigma cc^\top \Sigma e^{(A^\top -s  cc^\top \Sigma)t} dt,
 $$ 
 where we omit the argument $s$ in $\Sigma(s)$ in the above expression. The integral exists because $(A, c)$ is detectable, and hence $(A - s\Sigma(s) cc^\top)$ is Hurwitz.  It then follows that if $\Sigma'(s) < 0$, then $\Sigma(s) c\neq 0$, and hence $P \neq 0$. 
  \end{proof}

The statement of Lemma~\ref{lem:fixnorm} is not surprising. Indeed, the Euclidean norms of $b$ and $c$ can be thought as the actuation gain (e.g., specific impulse for spacecraft/rocket propulsion) and the signal-to-noise ratio (SNR), respectively. Increasing the actuation gain and/or the SNR yields a better performance, i.e., a smaller values of $\Phi$. We thus assume in the sequel that $|b|^2 = \epsilon$ and $|c|^2 = \delta$ are fixed positive numbers. We note that such an assumption is natural in system design as the actuation gain of the actuator and the SNR of the sensor are given, but  their placements/embedding in the control system will matter for the performance measure. 
 
With the preliminaries above, we now formulate the joint actuator-sensor design problem as follows:
\vspace{5pt}

\noindent
{\bf Joint actuator-sensor design problem.} Find a pair $(b,c)$ which minimizes 
$\Phi(b,c)$ under the constraint that $|b|^2 = \epsilon$ and $|c|^2 = \delta$ with $\epsilon, \delta > 0$.
\vspace{5pt}
 
We note here that unlike the LQG control problem, the optimal feedback control $u^*(t)$ and the minimum mean-squared-error estimate $\hat x(t)$ can be solved ``independently'', the arguments $b$ and $c$ in $\Phi(b,c)$ are coupled---they are coupled in the term $\tr(A^\top K\Sigma + \Sigma K A)$. Hence, the joint actuator-sensor design problem cannot be solved by dividing it into  subproblems of actuator and sensor design.  
 
\section{Double bracket flow as \\ a gradient descent algorithm}
We derive in the section a gradient descent algorithm which minimizes the potential function $\Phi(b,c)$. To introduce such a gradient descent algorithm, we need to first identify the solution space, and then impose a metric on the solution space. This is done in the first Subsection. We note here that similar computation and argument has been carried out in~\cite{belabbas2016geometric}. We thus omit a few computational details.       

\subsection{Solution space and normal metric}

{\em 1 Solution space.} To proceed, we first identify the underlying solution space. First, note that the PSD solutions $K(b)$ and $\Sigma(c)$ to~\eqref{eq:algebraicriccati} (and hence $\Phi(b,c)$) depend on $b$ and $c$ in a way such that they depend only on $bb^\top$ and $cc^\top$. Said more explicitly, if we normalize $bb^\top$ and $cc^\top$ as
\begin{equation}\label{eq:normailization}
B: =   bb^\top/\epsilon, \qquad C:= cc^\top / \delta 
\end{equation}
so that $\tr(B) = \tr(C) =1$, then we can re-write~\eqref{eq:algebraicriccati} as 
\begin{equation}\label{eq:algebraicriccati}
\left\{
\begin{array}{l}
 A^{\top} K + K A - \epsilon KBK + I = 0,   \\
A \Sigma + \Sigma A^{\top} - \delta \Sigma C  \Sigma + I = 0. 
\end{array}
\right. 
\end{equation} 
For the above reason, we can write $K(B)$ and $\Sigma(C)$, and hence $\Phi(B,C)$ without any ambiguity. The collection of such pair $(B,C)$ will be the solution space. Specifically, we let $\{e_i\}^n_{i = 1}$ be the standard basis of $\R^n$, and $$\SO(n):= \{\Theta \in \R^{n\times n} \mid \Theta^\top \Theta = I\}$$ be the special orthogonal group.  We then let 
\begin{equation}\label{eq:solutionspace}
X:= \left\{ \Theta e_i e_i^\top \Theta \mid \Theta \in \SO(n) \right\},
\end{equation}
to which the matrices $B$ and $C$ belong.  
The space $X$ is also known as a coadjoint orbit. Note that the above definition does not depend on the choice of $e_i$ as $\Theta$ acts transitively on the unit sphere $S^{n-1}$. Since $B,C\in X$, the solution space is then the product space $X^2:= X\times X$.

{\em 2. Normal metric}. A metric (tensor) $g$ on the space $X$ is such that at each point $B\in X$, $g_B$ is a positive definite bilinear form on $T_BX$ (the tangent space of $X$ at~$B$), and varies smoothly on $B$. Equipped with a metric $g$, $(X,g)$ is then a Riemannian manifold. 

For the coadjoint orbit $X$, there is a canonical metric, known as the normal metric, which will be characterized below. Let $$\so(n):= \{\Omega \in \R^{n\times n}\mid \Omega^\top + \Omega = 0 \}$$ be the set of skew-symmetric matrices.  Denote by $[\cdot,\cdot]$ the commutator of matrices. Then, the tangent space of $X$ at  a a matrix $B\in X$ is given by
\begin{equation}\label{eq:tangentspace}
T_BX= \{[B, \Omega] \mid  \Omega \in \so(n)\}. 
\end{equation}
Fix the matrix $B$, and let $\ad_B(\cdot) := [B, \cdot]$ be the linear map from $\so(n)$ to $T_B X$. The linear map is onto, and we denote by ${\rm ker_B}$
 the kernel of $\ad_B$. If one imposes the inner product on $\so(n)$ by $-\tr(\Omega\Omega')$, then the subspace of $\so(n)$ perpendicular to ${\rm ker}_B$ is defined. We denote it by ${\rm ker}^\perp_B$. So, ${\rm ker}_B\oplus {\rm ker}^\perp_B  = \so(n)$, and $\tr({\rm ker}_B {\rm ker}^\perp_B) = 0$. 
 
 It then follows that $\ad_B$, when restricted to ${\rm ker}^\perp_B$, is a linear isomorphism between ${\rm ker}^\perp_B$ and $T_BX$. We can thus re-write~\eqref{eq:tangentspace} as
 \begin{equation}\label{eq:tangentspace1}
 T_BX= \{\ad_B\Omega \mid  \Omega \in {\rm ker}^\perp_B\}.
 \end{equation}
The normal metric (tensor)  on $X$ is then defined as follows: 
\begin{equation}\label{eq:normalmetric}
g_B(\ad_B \Omega, \ad_B\Omega') := -\tr(\Omega \Omega')
\end{equation}
for $\Omega, \Omega'\in {\rm ker}^\perp_B$.

In the case here, we have $X^2$ the solution space. One can simply extend the normal metric to the product space $X^2$ as 
\begin{multline}\label{eq:normalmetric2}
g_{(B,C)}((\ad_B\Omega_B, \ad_C\Omega_C ),(\ad_B\Omega'_B, \ad_C\Omega'_C ) ) \\ := -\tr(\Omega_B\Omega'_B) -\tr(\Omega_C\Omega'_C) 
\end{multline}
for $\Omega_B, \Omega'_B\in \ker_B^\perp$ and $\Omega_C, \Omega'_C\in \ker_C^\perp$.

\subsection{Gradient descent algorithm} We derived here the gradient flow of $\Phi(B,C)$ over the solution space $X^2$ with respect to the normal metric~$g$ defined in~\eqref{eq:normalmetric2}. Denote by $\nabla \Phi$ the gradient of $\Phi$, determined by the following defining condition:
\begin{equation}\label{eq:defininggradient}
g_{(B,C)}(\nabla \Phi(B,C), v) = v \cdot \Phi(B,C), \quad \forall v \in T_{(B,C)} X^2, 
\end{equation}
where $v\cdot \Phi(B,C)$ is the directional derivative along~$v$. We are now in the a position to state the first main result:

\begin{theorem}\label{thm:mainresult}
Let the potential function $\Phi(B,C)$ be defined in~\eqref{eq:performance} over the solution space $X^2$ for $X$ defined in~\eqref{eq:solutionspace}. Let the metric $g$ be defined in~\eqref{eq:normalmetric2}. Then, the gradient flow 
$$(\dot B, \dot C) = - \nabla \Phi(B,C)$$ 
which minimizes $\Phi$ is given as follows:
\begin{equation}\label{eq:gradientflow}
\left\{
\begin{array}{l}
\dot B = \epsilon\delta [B,[B, KMK]], \\
\dot C = \epsilon \delta[C,[C,\Sigma N \Sigma ]],
\end{array}
\right. 
\end{equation}
where $K, \Sigma$ are the PSD solutions to the AREs:
\begin{equation*}
\left\{
\begin{array}{l}
 A^{\top} K + K A - \epsilon KBK + I = 0,   \\
A \Sigma + \Sigma A^{\top} - \delta \Sigma C  \Sigma + I = 0,
\end{array}
\right.  
\end{equation*}
and $M, N$ are the solutions to the Lyapunov equations:
\begin{equation}\label{eq:LyaEquation}
\left\{
\begin{array}{l}
(A - \epsilon BK) M + M (A - \epsilon BK)^{\top} +  \Sigma C \Sigma = 0,  \\
(A - \delta\Sigma C)^{\top} N + N (A - \delta\Sigma C) + K B K = 0.  
\end{array}
\right. 
\end{equation}\,
\end{theorem}

\begin{remark} Note that one can re-scale the gradient flow~\eqref{eq:gradientflow} by dividing $\epsilon \delta$:
\begin{equation}\label{model0}
\left\{
\begin{array}{l}
\dot B =  [B,[B, KMK]],\\
\dot C =  [C,[C,\Sigma N \Sigma ]],
\end{array}
\right. 
\end{equation}
In particular, the two dynamics share the same set of equilibria. 
\end{remark}

\begin{remark}\label{rmk:trueforBC}
We also note that if the initial conditions $B(0) = b(0)b(0)^\top/\epsilon$ and $C(0) = c(0)c(0)^\top/\delta$ are chosen such that $(A, b(0))$ is stabilizable and $(A, c(0))$ is detectable (which is generically true), then $(A, b(t))$ is stabilizable and $(A, c(t))$ is detectable for all $t\ge 0$, where $b(t)$ and $c(t)$ are such that $B(t) = b(t)b(t)^\top/\epsilon$ and $C(t) = c(t)c(t)^\top/\delta$. In particular, the matrices $M$ and $N$ in~\eqref{eq:LyaEquation} can be expressed as follows:
\begin{equation}\label{eq:integralformulaforMN}
\left\{
\begin{array}{l}
M = \displaystyle \int^\infty_0 e^{(A - \epsilon BK)t}\, \Sigma C\Sigma \, e^{(A - \epsilon BK)^\top t}dt, \vspace{3pt} \\
N = \displaystyle \int^\infty_0 e^{(A - \delta \Sigma C)^\top t}\, K B K \, e^{(A - \delta \Sigma C) t}dt. 
\end{array}
\right. 
\end{equation}\,
\end{remark}

We provide below a proof of Theorem~\ref{thm:mainresult}. 

\begin{proof}[Proof of Theorem~\ref{thm:mainresult}]
The proof follows from computation by matching the two sides of~\eqref{eq:defininggradient}. Fix a pair $(B,C)$ in the solution space $X^2$. Pick a tangent vector $v\in T_{(B,C)}X^2$ and we write $v = (v_B, v_C) $,  with $v_B = \ad_B \Omega_B$ and $v_C = \ad_C \Omega_C$ for $\Omega_B\in \ker_B^\perp$ and $\Omega_C \in \ker_C^\perp$. 

By computation, the directional derivative $v\cdot \Phi(B,C)$ is given by  
\begin{equation}\label{eq:ss}
v\cdot \Phi= \delta\tr(\Sigma C \Sigma K' ) + \epsilon \tr(K B K \Sigma'),
\end{equation}
where $K' := v_B \cdot K$ and $ \Sigma' := v_C\cdot \Sigma$ are directional derivatives,  which satisfy the following Lyapunov equations:
$$
\left\{
\begin{array}{l}
(A- \epsilon BK)^\top K' +  K' (A - \epsilon BK) -\epsilon  K [B, \Omega_B] K = 0,  \\
(A - \delta\Sigma C) \Sigma' + \Sigma' (A - \delta\Sigma C)^\top -\delta \Sigma [C, \Omega_C] \Sigma = 0.  
\end{array}
\right. 
$$
Since $(A,b)$ is stabilizable and $(A, c)$ is detectable, $(A - \epsilon BK)$ and $(A - \delta \Sigma C)$ are Hurwitz matrices, and hence $K'$ and $\Sigma'$ can be expressed as 
$$
\left\{
\begin{array}{l}
K' = -\epsilon\displaystyle \int^\infty_{0} e^{(A- \epsilon BK)^\top t} \, K [B,\Omega_B] K \, e^{(A- \epsilon BK) t} dt, \vspace{3pt}\\
\Sigma' = - \delta \displaystyle \int^\infty_0  e^{(A- \delta \Sigma C) t} \, \Sigma [C,\Omega_C] \Sigma \, e^{(A- \delta \Sigma C )^\top t} dt, 
\end{array}
\right.
$$
Combining the above integral formula with~\eqref{eq:ss} and using the fact that $\tr(A[B,C]) =\tr([A, B]C)$ for any square matrices $A$, $B$, and $C$, we obtain
\begin{equation}\label{eq:matching1}
v\cdot \Phi = \epsilon \delta (\tr( [B, KMK] \Omega_B ) + \tr([C, \Sigma N\Sigma] \Omega_C))
\end{equation}
with $M$ and $N$ given by~\eqref{eq:integralformulaforMN}, or equivalently by~\eqref{eq:LyaEquation}.

Now, for the left hand side of~\eqref{eq:defininggradient}, we first note that $\nabla\Phi(B,C)$ must take the form 
$$
\nabla \Phi(B,C) = ([B, \Omega^*_B],   [C, \Omega^*_C])
$$
with $\Omega^*_B\in \ker^\perp_B$ and $\Omega^*_C\in \ker^\perp_C$. This directly follows from~\eqref{eq:tangentspace1}. Thus, the gradient $\nabla\Phi$ will be determined if $\Omega^*_B$ and $\Omega^*_C$ are known to us. 
By the definition of the normal metric $g$ in~\eqref{eq:normalmetric2},
we have
\begin{equation}\label{eq:matching2}
g_{(B,C)}(\nabla\Phi(B,C), v) = - \tr(\Omega^*_B\Omega_B) - \tr(\Omega^*_C\Omega_C).
\end{equation}
By matching~\eqref{eq:matching1} and~\eqref{eq:matching1}, we obtain 
\begin{equation}\label{eq:omegaBomegaC}
\Omega^*_B = - [B, KMK], \quad \Omega^*_C = - [C, \Sigma N \Sigma]
\end{equation}
provided that $[B, KMK]$ and $[C,\Sigma N \Sigma]$ belong to $\ker^\perp_B$ and $\ker^\perp_C$, respectively. Note that if~\eqref{eq:omegaBomegaC} holds, then the proof is done.  

We now show that $[B, KMK]\in\ker^\perp_B$. The same argument can be applied to establish $[C,\Sigma N \Sigma]\in \ker^\perp_C$. Pick any $\Omega \in \ker_B$, i.e., $[B,\Omega] = 0$, we need to show that 
$\tr(\Omega [B, KMK]) = 0$. This follows because
$$
\tr(\Omega [B, KMK]) = -\tr([B,\Omega]KMK) = 0,
$$
which completes the proof.
\end{proof}

\section{Analysis of the gradient descent algorithm for symmetric, stable systems}
We call a pair $(B,C)\in X^2$ an {\em equilibrium point} of the gradient  $\nabla\Phi$ if  $\nabla\Phi(B,C) = 0$. Equivalently, an equilibrium point is also a {\em critical point} of the potential $\Phi$. An optimal solution $(B,C)$, i.e., a global minimum point of $\Phi$, is necessarily an equilibrium point of the gradient flow. Thus, characterizing equilibria (and especially, stable equilibria) is crucial.    Although the gradient descent algorithm~\eqref{eq:gradientflow} (or the re-scaled version~\eqref{model0}) can be applied to arbitrary control systems, the analysis of its equilibria can be quite difficult in general.   
In the section, we focus on a special class of control systems---these systems are such that the system matrix $A$ is negative definite with distinct eigenvalues, and the Euclidean norms of $|b|$ and $|c|$ are relatively small. 

The goal here is thus to demonstrate the type of analysis one needs to carry out for computing the set of equilibria, and provide insights into the analysis for a general case. 

We note  here that with a few more arguments, 
the results obtained here can be extended to the case where $A$ is negative {\em semi-}definite which, for example, includes the class of (weighted) Laplacian dynamics, i.e., 
\begin{equation}\label{eq:laplaciandynamics}
\left\{
\begin{array}{l}
dx_t = -L x_t dt + bu(t)dt + dw_t, \\
d y_t = c^\top x_t dt + d\nu_t.
\end{array}
\right.
\end{equation}
where $L = [l_{ij}]$ a weighted irreducible Laplacian matrix, i.e., $l_{ij} = l_{ji}\ge 0$ and $\sum^n_{j = 1}l_{ij} = 0$ for all $i = 1,\ldots, n$. But for ease of exposition, we focus only on the case where $A$ is negative definite. For the case where $A$ is unstable, the analysis becomes more subtle, and the details will be discussed in a different paper.

\subsection{The eigenvector problem}
 For given positive numbers $\epsilon$ and $\delta$, we denote by $E_{(\epsilon, \delta)}$ the set of equilibria associated with~\eqref{eq:gradientflow}. Note that the dynamics~\eqref{model0} shares the same set of equilibria with~\eqref{eq:gradientflow}, except for the case where either $\epsilon$ or $\delta$ is zero. Indeed, if $\epsilon\delta = 0$, then $\nabla\Phi(B,C) = 0$ for all $(B,C)$, which does not hold for~\eqref{model0}. We thus have $(0,0)$ the point of singularity. On the other hand, 
 one may treat $E_{(\epsilon, \delta)}$ as the set of equilibria associated with~\eqref{model0}. In this way, $E_{\epsilon, \delta}$ is defined for all nonnegative $\epsilon$ and $\delta$. The benefits for one to do this is the following: Note that~\eqref{model0} depends smoothly on $\epsilon$ and $\delta$ (the dependence is via $K$, $\Sigma$, $M$, and $N$). Thus, by arguments of perturbation, one would expect that the set $E_{(\epsilon, \delta)}$ also varies smoothly over an open neighborhood of $(0,0)$ provided that equilibria satisfy certain non-degenerateness conditions. 

We now characterize conditions for a pair $(B,C)$ to be an equilibrium point. By definition, we have
$$
\left\{
\begin{array}{l}
 [B,[B, KMK]] = 0,\\
\left [C,[C,\Sigma N \Sigma ] \right ] = 0.
\end{array}
\right. 
$$
On the other hand, we have shown in the proof of Theorem~\ref{thm:mainresult} that $[B, KMK]$ and $[C, \Sigma N \Sigma]$ belong to $\ker^\perp_B$ and $\ker^\perp_C$, respectively. Since $\ker^\perp_B$ and $\ker^\perp_C$ are linearly isomorphic to $T_BX$ and $T_C X$, respectively, we have
$$
\left\{
\begin{array}{l}
 [B, KMK] = 0,\\
\left [C,\Sigma N \Sigma \right ]  = 0.
\end{array}
\right. 
$$
Conversely, if the above equations hold, then $(B,C)$ is an equilibrium point of $\nabla\Phi(B,C)$.

It is well known that two matrices commute if and only if they share the same set of eigenvectors. Now let $\bar b, \bar c\in \R^n$ be defined such that $B = \bar b\bar b^\top$ and $C = \bar c\bar c^\top$, it follows from the above commutators that 
\begin{equation}\label{eq:eigenvectorproblem}
\left\{
\begin{array}{l}
 KMK \bar b = \mu_B \bar b,\\
\Sigma N \Sigma \bar c  = \mu_C \bar c,
\end{array}
\right. 
\end{equation}
for $\mu_B, \mu_C\in \R$, i.e., $\bar b$ and $\bar c$ are eigenvectors of $KMK$ and $\Sigma N \Sigma$, respectively. Note that  optimal actuator and sensor vectors  $b$ and $c$ can be obtained as $b = \sqrt{\epsilon}\, \bar b$  and $c = \sqrt{\delta}\, \bar c$.   
The above equation serves as the starting point of our analysis.  

\subsection{On the case where $\epsilon = \delta = 0$}  The eigenvector problem posed in~\eqref{eq:eigenvectorproblem} is in general hard to solve. The difficulty lies in the fact that $K$, $\Sigma$, $M$, and $N$ are nonlinear in~$B$ and $C$ (and hence $\bar b$ and $\bar c$). Yet, such nonlinearity vanishes if $\epsilon = \delta = 0$; indeed, we have in this case the following sets of equations:
\begin{equation*}
\left\{
\begin{array}{l}
 A^{\top} K + K A + I = 0,   \\
A \Sigma + \Sigma A^{\top} + I = 0,
\end{array}
\right.  
\end{equation*}
and 
\begin{equation}\label{eq:defMandN}
\left\{
\begin{array}{l}
A  M + M A^{\top} +  \Sigma C \Sigma = 0,  \\
A^\top N + N A  + K B K = 0.  
\end{array}
\right. 
\end{equation} 
Since $A$ is symmetric, $K$ and $\Sigma$ satisfy the same equation and can be solved explicitly as
$$
K = \Sigma = - \frac{1}{2} A^{-1}
$$   
Now, let $A = \Theta \Lambda \Theta^\top$, with $\Theta$ an orthogonal matrix and $\Lambda =\diag (\lambda_1, \cdots, \lambda_n)$ a diagonal matrix. For convenience, we define vectors $\beta$ and $\gamma$ as follows: 
\begin{equation}\label{eq:defbetagamma}
\beta:= \Theta^\top K \bar b = \Lambda^{-1} \Theta^\top \bar b, \quad  \gamma := \Theta^\top \Sigma\bar c = \Lambda^{-1} \Theta^\top \bar c.  
\end{equation}
The normalization condition for $\beta$ and $\gamma$ is such that 
\begin{equation}\label{eq:normalization}
\beta^\top \Lambda^2 \beta = \gamma^\top \Lambda^2\gamma  = 1.
\end{equation}
We then solve $M$ and $N$, using the newly defined variables $\beta$ and $\gamma$, as follows: 
$$
\left\{
\begin{array}{l}
M = \frac{1}{4}  \Theta  \diag(\gamma)\Psi \diag(\gamma)  \Theta^\top \vspace{3pt} \\
N = \frac{1}{4} \Theta \diag(\beta)\Psi \diag(\beta) \Theta^\top, 
\end{array}
\right.
$$
where $\diag(\beta)$ and $\diag(\gamma)$ are diagonal matrices with $\beta$ and $\gamma$ on their diagonals, and $\Psi$ is a positive-definite Cauchy matrix~\cite{schechter1959inversion} (note that the $\lambda_i$'s are negative) given by
$$
\Psi := \left [-\frac{1}{\lambda_i + \lambda_j} \right ]_{ij}. 
$$
So, with the above closed expressions, we can re-write~\eqref{eq:eigenvectorproblem} as follows: 
\begin{equation}\label{eq:betagammaequations}
\left\{
\begin{array}{l}
 \diag(\gamma) \Psi \diag(\gamma) \beta = \mu_\beta \Lambda^2 \beta \\
  \diag(\beta) \Psi \diag(\beta) \gamma = \mu_\gamma \Lambda^2 \gamma, 
\end{array}
\right. 
\end{equation}
with $\mu_\beta = 4 \mu_B$ and $\mu_\gamma = 4 \mu_C$. 

\begin{remark} {\em We note here that a similar set of equations has been investigated in~\cite{chen2017optimalact}. Yet, the results there cannot be simply applied here. Specifically, the problem  addressed in~\cite{chen2017optimalact} is the following: Consider a deterministic linear control system $\dot x = -\Lambda x + bu$ with $\Lambda$ the diagonal matrix defined above (thus, $-\Lambda$ is positive definite).  The minimal energy consumption for driving the system from an initial condition $x\in S^{n-1}$ in the unit sphere to the origin is given by $\eta(b,x) = x^\top\diag(b)^{-1}\Psi^{-1}\diag(b)^{-1} x$. We posed and solved in~\cite{chen2017optimalact} the minimax problem $\min_{b\in S^{n-1}}\max_{x\in S^{n-1}}\eta(b,x)$. 
In particular, a necessary and sufficient condition for a pair $(b,x)$ to be a minimax solution is such that  
\begin{equation}\label{eq:bax}
\left\{
\begin{array}{l}
 \diag(x) \Psi \diag(x) b = \mu b, \\
  \diag(b) \Psi \diag(b) x = \mu x, 
\end{array}
\right. 
\end{equation}
where $b$ and $x$ do {\em not} have any single entry, and moreover, $\mu$ is the {\em smallest} eigenvalue of $\diag(b)\Psi\diag(b)$.  We note that although the two sets of equations~\eqref{eq:betagammaequations} and~\eqref{eq:bax} are similar, the problems are different; indeed, as we will see that the optimal solution $(\beta,\gamma)$ is such that $\beta$ and $\gamma$ can have  {\em only one} nonzero entry.        
}
\end{remark}

\subsection{The set of equilibria.}
We now characterize solutions to~\eqref{eq:betagammaequations}, which one-to-one correspond to the points in the set of equilibria $E_{(0,0)}$ via~\eqref{eq:defbetagamma}.  
We will see soon that $E_{(0,0)}$ can be divided into two subsets---one subset can be realized as the zeros of certain algebraic equations, and hence is an algebraic variety. Moreover, the points in this subset are global maxima of a potential function (over $X^2$) whose gradient flow is given by~\eqref{model0}, and hence are unstable. The other set is  comprised only of isolated points,  and contains a unique stable equilibrium.

We introduce the following notation: For  vectors $\beta$ and $\gamma$, we define $\cal{I}_\beta$ and $\cal{I}_\gamma$, as subsets of $\{1,\ldots, n\}$, to be the collections of  indices of nonzero entries of $\beta$ and $\gamma$, respectively.   

\begin{proposition}\label{lbg1}
Suppose $(\beta,\gamma)$ is a solution to \eqref{eq:betagammaequations}, then either $\cal{I}_{\beta} \cap \cal{I}_{\gamma} = \varnothing$ or $\cal{I}_{\beta} = \cal{I}_{\gamma}$. 
\end{proposition}

\begin{proof}
First, note that if $\cal{I}_{\gamma} \cap \cal{I}_{\beta} = \varnothing$, then 
$
\diag(\beta) \gamma = \diag(\gamma) \beta  = 0
$, and hence 
\begin{equation*}
  \diag(\beta) \Psi \diag(\beta)\gamma = \diag(\gamma) \Psi \diag(\gamma) \beta = 0,
\end{equation*}
which implies that $(\beta,\gamma)$ is a solution. 

We now assume $\cal{I}_{\gamma} \cap \cal{I}_{\beta} \neq \varnothing$ and prove $\cal{I}_{\beta} = \cal{I}_{\gamma}$. 
We first show that $\cal{I}_{\gamma} \subseteq \cal{I}_{\beta}$. The proof will be carried out by contradiction. Without loss of generality, we assume that the first $k$ entries of $\beta$ are nonzero. Then 
\begin{equation}
\diag(\beta) \Psi \diag(\beta) = 
\begin{bmatrix}
\diag(\beta') \Psi'\diag(\beta') & 0 \\
0 & 0
\end{bmatrix}
\end{equation}
where $\beta'$ is comprised of the first $k$ entries of $\beta$ and $\Psi'$ is the associated leading principal sub-matrix of $\Psi$. Partition $\Lambda = \diag(\Lambda',\Lambda'')$ and $\gamma$ correspondingly, we then have
\begin{equation}\label{teq1111}
\begin{bmatrix}
\diag(\beta') \Psi'\diag(\beta')  \gamma' \\
0
\end{bmatrix} 
=  \mu_\gamma
\begin{bmatrix}
\Lambda'^2 \gamma' \\
\Lambda''^2 \gamma''
\end{bmatrix} 
\end{equation}
Now, suppose that $\cal{I}_{\gamma} \subsetneq \cal{I}_{\beta}$; then,  $\gamma''\neq 0$, and hence $\Lambda''^2\gamma'' \neq 0$.  But then, from~\eqref{teq1111}, we have $\mu_\gamma = 0$, and hence $\diag(\beta') \Psi'\diag(\beta')  \gamma' = 0$. On the other hand, the matrix $\diag(\beta') \Psi'\diag(\beta') $ is positive definite since $\Psi'$ is (see~\cite{chen2017optimalact}). So, we must have $\gamma' = 0$, and hence $\cal{I}_{\gamma} \cap\cal{I}_{\beta} = \varnothing $ which is a contradiction. It thus follows that $\cal{I}_{\gamma} \subseteq \cal{I}_{\beta}$. Now, we apply the same arguments but exchange the roles of  $\beta$ and $\gamma$, and obtain that $\cal{I}_{\beta} \subseteq \cal{I}_{\gamma}$. We thus conclude that $\cal{I}_{\gamma} =\cal{I}_{\beta} $.
\end{proof}

The subset of pairs $(\beta,\gamma)$ with $\cal{I}_\beta \cap \cal{I}_\gamma = \varnothing$ can be characterized by the following algebraic equations: 
\begin{equation}\label{eq:algebraicset}
\sum^n_{i=1}\lambda^2_i\beta^2_i = \sum^n_{i=1}\lambda^2_i\gamma^2_i =1 \mbox{ and } \sum^n_{i=1}(\beta_i \gamma_i)^2 = 0,\\
\end{equation} 
where the first equation comes from~\eqref{eq:normalization} and second comes from $\cal{I}_\beta \cap \cal{I}_\gamma = \varnothing$. 
We now show that any equilibrium $(B,C)$ corresponding a point in such subset is unstable under the dynamics~\eqref{model0}. We have the following fact:

\begin{proposition}\label{lem:newpotentialfunction}
Consider the following potential function: 
$$\overline{\Phi}(B,C) := \tr(A^\top M N + NM A),$$
with $M$ and $N$ defined in~\eqref{eq:defMandN} (the dependence of $B$ and $C$ is via $N$ and $M$, respectively). Then,~\eqref{model0} is the associated gradient flow with respect to the normal metric. Moreover, $\overline \Phi(B,C) \le 0$ for all $(B,C)\in X^2$, and $\overline\Phi(B,C) = 0$ if $\cal{I}_\beta \cap \cal{I}_\gamma = \varnothing$.  
\end{proposition}

\begin{proof}
We omit the proof that~\eqref{model0} is the gradient flow of $\overline\Phi$. It directly follows from computation, and the derivation is similar to the proof of Theorem~\ref{thm:mainresult}. 

We show that $\overline\Phi \le 0$. Since $KBK \ge 0$, $\Sigma C \Sigma \ge 0$, and $A$ is stable,  from the Lyapunov equation~\eqref{eq:defMandN}, we have $M\ge 0$ and $N \ge 0$. So,
$$
\overline\Phi(B,C) = -\tr(KBK M) = -\tr(\Sigma C \Sigma N) \le 0.
$$ 
Now, let $(B,C)$ be such that $\cal{I}_\beta \cap \cal{I}_\gamma = 0$. Without loss of generality, we assume that the first $k$ entries of $\beta$ are nonzero. Then, the zero patterns of $M$ and $N$ are given by:
$$
M = 
\begin{bmatrix}
0 & 0 \\
0 & *
\end{bmatrix} \mbox{ and } 
N = 
\begin{bmatrix}
* & 0 \\
0 & 0
\end{bmatrix},
$$
and hence $MN = 0$, which implies that $\overline \Phi = 0$. 
\end{proof}

For the remainder of the subsection, we focus on the case where $\cal{I}_\beta = \cal{I}_\gamma$. We fix a {\em nonempty} subset $\cal{I}'$ of $\{1,\ldots, n\}$, and assume that $\cal{I}_\beta = \cal{I}_\gamma = \cal{I}'$. Without loss of generality, we assume that $\cal{I}' = \{1,\ldots, k\}$, for $k\le n$.

Further, we denote by $S_{\cal{I}'}$ a finite abelian group defined as follows:
$$
S_{\cal{I}'} := \{ (d_1, \ldots, d_k, {\bf 0}_{n - k})\in \R^n \mid d^2_i = 1 \}.
$$  
Let ``$*$'' be the Hadamard product (i.e., entry-wise multiplication). Then, it should be clear that if $s, s'\in S_{\cal{I}'}$, then $s*s' = s'*s$, with $$s_{\rm id}:=({\bf 1}_k,{\bf 0}_{n - k})$$ the identity element.  In particular, $s*s = s_{\rm id}$ for all $s\in S_{\cal I'}$

The group $S_{\cal{I}'}$ acts on the pair of vectors $(\beta,\gamma)$ by $$s\cdot (\beta, \gamma) := (s*\beta, s*\gamma)$$ for any $s\in S_{\cal{I}'}$. One of the main purposes of introducing the abelian group $S_{\cal{I}'}$ is the following:

\begin{lemma} If $(\beta, \gamma)$, with $\cal{I}_\beta = \cal{I}_\gamma = \cal{I}'$, is a solution to~\eqref{eq:betagammaequations}, then $s\cdot (\beta, \gamma)$ is a solution to~\eqref{eq:betagammaequations} as well. 
\end{lemma}

We omit the proof as it follows from computation.  For a pair $(\beta, \gamma)$,  we denote by $O_{(\beta, \gamma)}$ the orbit under the group action, i.e., 
$$
O_{(\beta, \gamma)}:= \{s\cdot (\beta, \gamma) \mid s\in S_{\cal{I}'} \}
$$
We further note that the potential function $\overline{\Phi}$ defined in Lemma~\ref{lem:newpotentialfunction} is invariant under the group action. Specifically, let $(\bar b\bar b^\top ,\bar c\bar c^\top)$ and $(\bar b'\bar b'^\top,\bar c'\bar c'^\top)$ be two pairs in $X^2$, and let $(\beta, \gamma)$ and $(\beta',\gamma')$ be defined by~\eqref{eq:defbetagamma}. If $(\beta',\gamma') = s\cdot (\beta, \gamma)$, then $\overline \Phi(\bar b\bar b^\top ,\bar c\bar c^\top) = \overline \Phi(\bar b'\bar b'^\top,\bar c'\bar c'^\top)$. We omit the computational details. It follows that if an equilibrium $(\beta, \gamma)$ is stable/unstable under~\eqref{model0}, then so is any pair in its orbit.

We will now state facts about solutions $(\beta,\gamma)$ to~\eqref{eq:betagammaequations} with $\cal{\beta} = \cal{\gamma} = \cal{I}'$. Let $\lambda:= (\lambda_1,\ldots, \lambda_n)\in \R^n$.  For a given $s\in S_{\cal{I}'}$, we define a vector $\xi_{s}$ as follows: 
$$
\xi_s:= (\diag(s) \Pi \diag(s))^+ (s_{\rm id}*\lambda  * \lambda)
$$
where $M^+$ is the MooreÐPenrose inverse of~$M$.  
In the case here,  we have that  $M = [M', 0; 0, 0]$ is a symmetric matrix with $M'$ nonsingular. Then, $M^+ = [M'^{-1}, 0; 0, 0]$.

We introduce a few notations here. For a vector $v = (v_1,\ldots, v_n)\in \R^n$, we let $$\sgn(v):= (\sgn(v_1),\ldots, \sgn(v_n))$$ where $\sgn(\cdot)$ is the sign function. 
we write $v > 0$  (resp. $v\ge 0$) if each entry of $v$ is nonnegative (resp. positive). Furthermore, for any vector $v \ge 0$,  we let $$\sqrt{v} := (\sqrt{v_1},\ldots, \sqrt{v_n}).$$ 
With the notations above, we state the following fact: 

\begin{proposition}\label{prop:finitemany}
Let $\cal{I}'$ be a nonempty subset of $\{1,\ldots, n\}$. 
Then, for each $s\in S_{\cal{I}'}$, there exists at most one orbit $O_{(\beta,\gamma)}$, with $$\cal{I}_\beta = \cal{I}_\gamma = \cal{I}' \quad \mbox{ and }\quad \sgn(\beta*\gamma) = s,$$ such that $(\beta, \gamma)$ is a solution to~\eqref{eq:betagammaequations}. Moreover, such an orbit exists if and only if $\xi_s > 0$ and $(\beta, \gamma)$ can be chosen such that 
\begin{equation}\label{eq:solutionsss}
(\beta, \gamma) = \frac{1}{|\Lambda \sqrt{\xi_s}|}\left (\sqrt{\xi_s}, s * \sqrt{\xi_s} \right ).
\end{equation}\,
\end{proposition}

\begin{remark}
Prop.~\ref{prop:finitemany} implies that there are finitely many solutions $(\beta, \gamma)$ to~\eqref{eq:betagammaequations} under the condition that $\cal{I}_\beta = \cal{I}_\gamma$. Indeed, there are only finite many subsets $\cal{I}'$ of $\{1,\ldots, n\}$ and finitely many elements  $s$ in $S_{\cal{I}'}$. Then, for fixed $\cal{I}'$ and $s\in S_{\cal{I}'}$, Prop.~\ref{prop:finitemany} says that there is at most one orbit $O_{(\beta, \gamma)}$, with $\cal{I}_\beta = \cal{I}_\gamma = \cal{I}'$ and $\sgn(\beta*\gamma) = s$, which is comprised of as many as $2^{|\cal{I}'|}$ elements. 
\end{remark}

\begin{proof}[Proof of Prop.~\ref{prop:finitemany}]
We first show that $\beta = s* \gamma$. Let 
\begin{equation*}
w := \Psi \diag(\gamma)\beta = \Psi \diag(\beta )\gamma  
\end{equation*}
Then, we can re-write~\eqref{eq:betagammaequations} as follows:
\begin{equation}\label{teq2}
\left\{
\begin{array}{l}
\diag(w) \gamma = \mu_\beta \Lambda^2 \beta, \\ 
\diag(w) \beta = \mu_\gamma \Lambda^2 \gamma.\\
\end{array}
\right.
\end{equation}
Since $\mu_\beta$ and $\mu_\gamma$ are nonzero, we have $\cal{I}_{\beta} = \cal{I}_\gamma \subseteq \cal{I}_w$. It thus follows that
\begin{equation*}
\mu_\gamma \gamma_i^2 = \mu_\beta \beta^2_i, \quad \forall i = 1,\ldots, n.
\end{equation*}
From the normalization condition~\eqref{eq:normalization},  we obtain $\mu_\beta = \mu_\gamma$,   and hence $\beta = s* \gamma$. 

We now show that any such pair $(\beta, \gamma)$ satisfies the condition that $\beta*\beta$ (or $\gamma*\gamma$) is linearly proportional to the vector $\xi_s$.
From~\eqref{eq:betagammaequations}, we have
$$
\left\{
\begin{array}{l}
\diag(\beta)^+\diag(\gamma)\Psi\diag(\gamma) \beta = \mu_\beta \Lambda^2\diag(\beta)^+\beta, \\
\diag(\gamma)^+\diag(\beta)\Psi\diag(\beta) \gamma = \mu_\beta \Lambda^2\diag(\gamma)^+\gamma. 
\end{array}
\right.
$$
Using the fact that $\gamma = s*\beta$ and $s*s = s_{\rm id}$, we obtain 
$$
\left\{
\begin{array}{l}
\diag(s)\Psi\diag(s) (\beta*\beta) = \mu_\beta (s_{\rm id}*\lambda *\lambda), \\
\diag(s)\Psi\diag(s) (\gamma* \gamma) = \mu_\gamma (s_{\rm id}*\lambda *\lambda). 
\end{array}
\right.
$$
It thus follows that
$$
\beta*\beta = \gamma* \gamma \propto  \xi_s.
$$
The normalization condition~\eqref{eq:normalization} then yields~\eqref{eq:solutionsss}.
\end{proof}

\subsection{The unique stable equilibrium}
We have so far characterized the set of equilibria associated with~\eqref{model0} for the case $\epsilon = \delta = 0$. In particular, we have shown that any equilibrium $(B,C)$ with $\cal{I}_\beta \cap \cal{I}_\gamma = \varnothing$ is unstable.  
Recall that $0 > \lambda_1\ldots > \lambda_n$. We denote by $v_i$, with $|v_i| = 1$, an eigenvector of $A$ corresponding to $\lambda_i$.  We now state the following fact about the stable equilibrium of~\eqref{model0}:

\begin{theorem}\label{thm:iavsspeciallastgameforia??}
There is a unique stable equilibrium associated with~\eqref{model0} for the case $\epsilon = \delta = 0$, and is given by $$(B,C)= (v_1v^\top_1, v_1v_1^\top).$$\, 
\end{theorem}
  
 \noindent 
{\em Sketch of proof.} Due to space limitation, we provide below a sketch of the proof.

First note that $(v_1v^\top_1, v_1v_1^\top)$ is an equilibrium associated with~\eqref{model0}. Indeed, let $\bar b = \bar c = v_1$, then the corresponding $\beta$ and $\gamma$ are given by $\beta = \gamma = e_1 / \lambda_1$. The pair $(\beta,\gamma)$ is a solution to~\eqref{eq:betagammaequations}.  

The analysis of stability relies on computing the {\em Hessian} of the  potential function $\overline \Phi$ at a critical point $p = (B, C)\in X^2$. Roughly speaking, the Hessian at $p$, denoted by $H_{p}$, is a symmetric bilinear form on $T_p X^2$, defined as follows:
$$
H_p(v, w) := w\cdot v\cdot \Phi(B,C).
$$ 
More precisely, one needs to extend $w$ and $v$ locally to vector fields over an open neighborhood of $p$ to make the above definition works. But since $(B, C)$ is a critical point, the way how we extend $v$ and $w$ to get the local vector fields does not matter. When the state space is an Euclidean space, the Hessian can be represented as a symmetric matrix, known as the Hessian matrix $\partial^2 \Phi / \partial p^2$, and we simply have $H_p( v, w) = v^\top \partial^2 \Phi / \partial p^2 w$. 

A critical point $p$ is called {\em nondegenerate} if the bilinear form $H_p$ is nondegenerate, and is {\em exponentially stable} if it is positive definite. We evaluate the Hessian $H_p$ at any critical point. We show that $H_p$ is positive definite if and only if $p = (v_1v^\top_1, v_1v_1^\top)$, and moreover, $H_p$ is the only positive semi-definite Hessian among all the other $H_{p'}$'s for $p'$ a critical point. Thus, we conclude that $p$ is the unique (exponentially) stable equilibrium.  \hfill{\qed}

It is known that if an equilibrium is exponentially stable under a nominal dynamics, then it is robust under perturbation (of the dynamics) in the sense that there will be another exponentially stable equilibrium, close to the original one, associated with the perturbed dynamics. This, in particular, implies the following fact:

\begin{corollary}
The pair $(v_1v_1^\top, v_1v_1^\top)$ is the unique stable equilibrium associated with~\eqref{model0} for small $\epsilon$ and $\delta$. The minimum value of $\Phi(B, C)$ over $X^2$ is given by
\begin{multline*}
\min_{(B, C)\in X^2} \Phi(B,C) = \Phi(v_1v_1^\top, v_1v_1^\top) = \\
\frac{\sqrt{\lambda^2_1 + \epsilon} + \sqrt{\lambda^2_1 + \delta}}{(\sqrt{\lambda^2_1 + \epsilon} - \lambda_1)(\sqrt{\lambda^2_1 + \delta} - \lambda_1)} - \sum^n_{k = 2}\frac{1}{2\lambda_k}.
\end{multline*}\,
\end{corollary}

\begin{proof}
We show that $(v_1^\top v_1, v_1v_1^\top)$ is an equilibrium of~\eqref{model0}, or equivalently, $v_1$ is an eigenvector of $KMK$ and of $\Sigma N\Sigma$ by~\eqref{eq:eigenvectorproblem}. Recall that $A = \Theta \Lambda \Theta^\top$, and $\Theta^\top v_1 = e_1$.  
It suffices to show that $\Theta K \Theta^\top$, $\Theta M\Theta^\top$, $\Theta \Sigma\Theta^\top$, and $\Theta N\Theta^\top$ are all diagonal matrices.  This directly follows from computation. 
Note that $0> \lambda_1 \ldots > \lambda_n$. We have
\begin{equation*}\label{eq:KMSNdiagonal}
\left\{
\begin{array}{l}
 \Theta K\Theta^\top = \diag(\nicefrac{1}{\sqrt{\lambda^2_1 + \epsilon} - \lambda_1}, -\nicefrac{1}{2\lambda_2}, \ldots, -\nicefrac{1}{2\lambda_n}), \vspace{3pt}\\
\Theta \Sigma \Theta^\top= \diag(\nicefrac{1}{\sqrt{\lambda^2_1 + \delta} - \lambda_1}, -\nicefrac{1}{2\lambda_2}, \ldots, -\nicefrac{1}{2\lambda_n}), \vspace{3pt}\\
 \Theta M \Theta^\top= \diag(\nicefrac{1}{2\sqrt{\lambda^2_1+\epsilon} \left (\sqrt{\lambda^2_1 + \delta} - \lambda_1\right )^2}, {\bf 0}_{n-1}), \vspace{3pt}\\
 \Theta N\Theta^\top = \diag(\nicefrac{1}{2\sqrt{\lambda^2_1+\delta} \left (\sqrt{\lambda^2_1 + \epsilon} - \lambda_1\right )^2}, {\bf 0}_{n-1}). 
\end{array} 
\right.
\end{equation*} 
The fact that $(v_1^\top v_1, v_1v_1^\top)$ is the unique stable equilibrium (for $\epsilon$ and $\delta$ small) follows from the fact that the Hessian $H_p$ at a critical point $p$ depends smoothly on $\epsilon$ and $\delta$.  Finally, the minimum value of $\Phi$ directly follows from computation.   
\end{proof}

\begin{remark}
We note here that the above Corollary also applies to the case where the matrix $A$ is negative semi-definite. In particular, if we let $A= -L$ with $L$ a weighted irreducible Laplacian matrix, then $v_1 = {\bf 1}_n/\sqrt{n}$. Thus, an optimal actuator-sensor pair to the dynamics~\eqref{eq:laplaciandynamics} is  $(\sqrt{\nicefrac{\epsilon}{n}}{\bf 1}_n, \sqrt{\nicefrac{\delta}{n}}{\bf 1}_n )$.  
\end{remark}

\section{Conclusions}
We formulate in the paper the joint actuator-sensor design problem for stochastic linear systems. The goal is to optimize the actuator-sensor pair $(b,c)$ so as to minimize a time-averaged quadratic cost function. The closed formula for the cost function $\Phi$ is given in~\eqref{eq:performance}, the dependence of $b$ and $c$ is via the AREs for $K$ and $\Sigma$ (see~\eqref{eq:algebraicriccati}). We argued that $\Phi$ depends only on $B = bb^\top$ and $C = cc^\top$, and derived the corresponding gradient algorithm for minimizing $\Phi$ (see Theorem~\ref{thm:mainresult}) over the solutions space $X^2$. We then characterized the set of equilibria associated with the gradient flow under the assumption that $A$ is negative definite and $|b|$ and $|c|$ are relatively small. In the extreme case where $\epsilon = \delta = 0$,  we illustrated the type of analysis one needs to carry out for computing the set of equilibria, which involves solving the algebraic equations~\eqref{eq:betagammaequations}. The results are summarized in Props. 1-3. Further, we show that there is a unique stable equilibrium of the gradient flow under the above assumption (see Theorem~\ref{thm:iavsspeciallastgameforia??}).  A sketch of the proof  is given, but with computational details of the Hessians omitted. 

The joint sensor-actuator design problem formulated here is far from being completely solved. Nevertheless, we believe that the approach established here can be modified and extended to other cases.  Ongoing research includes the case where $A$ is Hurwitz but not necessarily symmetric, the case where $A$ is positive-definite, and the case where $b$ and $c$ are matrices comprised of multiple actuator/sensor vectors.

\bibliographystyle{unsrt}
\bibliography{paperrefs}

\begin{thebibliography}{10}

\bibitem{brockett1995stabilization}
R.W. Brockett.
\newblock Stabilization of motor networks.
\newblock In {\em Decision and Control, 1995., Proceedings of the 34th IEEE
  Conference on}, volume~2, pages 1484--1488. IEEE, 1995.

\bibitem{zhang2006communication}
L.~Zhang and D.~Hristu-Varsakelis.
\newblock Communication and control co-design for networked control systems.
\newblock {\em Automatica}, 42(6):953--958, 2006.

\bibitem{he2006sensor}
Y.~He and E.K.P. Chong.
\newblock Sensor scheduling for target tracking: {A} {M}onte {C}arlo sampling
  approach.
\newblock {\em Digital Signal Processing}, 16(5):533--545, 2006.

\bibitem{shi2013optimal}
D.~Shi and T.~Chen.
\newblock Optimal periodic scheduling of sensor networks: {A} branch and bound
  approach.
\newblock {\em Systems \& Control Letters}, 62(9):732--738, 2013.

\bibitem{han2017optimal}
D.~Han, J.~Wu, H.~Zhang, and L.~Shi.
\newblock Optimal sensor scheduling for multiple linear dynamical systems.
\newblock {\em Automatica}, 75:260--270, 2017.

\bibitem{chen2017optimal}
X.~Chen, M.-A. Belabbas, and T.~Ba{\c{s}}ar.
\newblock Optimal capacity allocation for sampled networked systems.
\newblock {\em Automatica}, 85:100--112, 2017.

\bibitem{belabbas2016geometric}
M.-A. Belabbas.
\newblock Geometric methods for optimal sensor design.
\newblock {\em Proc. R. Soc. A}, 472(2185):20150312, 2016.

\bibitem{chen2017optimalact}
X.~Chen and M.-A. Belabbas.
\newblock Optimal actuator design for minimizing the worst-case control energy.
\newblock {\em IFAC-PapersOnLine}, 50(1):9991--9996, 2017.

\bibitem{chen2014fluid}
K.K. Chen and C.W. Rowley.
\newblock Fluid flow control applications of {H}2 optimal actuator and sensor
  placement.
\newblock In {\em American Control Conference (ACC), 2014}, pages 4044--4049.
  IEEE, 2014.

\bibitem{hiramoto2000optimal}
K.~Hiramoto, H.~Doki, and G.~Obinata.
\newblock Optimal sensor/actuator placement for active vibration control using
  explicit solution of algebraic {R}iccati equation.
\newblock {\em Journal of Sound and Vibration}, 229(5):1057--1075, 2000.

\bibitem{rao1991optimal}
S.S. Rao, T.-S. Pan, and V.B. Venkayya.
\newblock Optimal placement of actuators in actively controlled structures
  using genetic algorithms.
\newblock {\em AIAA journal}, 29(6):942--943, 1991.

\bibitem{summers2016submodularity}
T.H. Summers, F.L. Cortesi, and J.~Lygeros.
\newblock On submodularity and controllability in complex dynamical networks.
\newblock {\em IEEE Transactions on Control of Network Systems}, 3(1):91--101,
  2016.

\bibitem{pasqualetti2014controllability}
F.~Pasqualetti, S.~Zampieri, and F.~Bullo.
\newblock Controllability metrics, limitations and algorithms for complex
  networks.
\newblock {\em IEEE Transactions on Control of Network Systems}, 1(1):40--52,
  2014.

\bibitem{tzoumas2016sensor}
V.~Tzoumas, A.~Jadbabaie, and G.J. Pappas.
\newblock Sensor placement for optimal {K}alman filtering: {F}undamental
  limits, submodularity, and algorithms.
\newblock In {\em American Control Conference (ACC), 2016}, pages 191--196.
  IEEE, 2016.

\bibitem{olshevsky2014minimal}
A.~Olshevsky.
\newblock Minimal controllability problems.
\newblock {\em IEEE Transactions on Control of Network Systems}, 1(3):249--258,
  2014.

\bibitem{zhang2017sensor}
H.~Zhang, R.~Ayoub, and S.~Sundaram.
\newblock Sensor selection for kalman filtering of linear dynamical systems:
  {C}omplexity, limitations and greedy algorithms.
\newblock {\em Automatica}, 78:202--210, 2017.

\bibitem{kalman1961new}
R.~E. Kalman and R.~S. Bucy.
\newblock New results in linear filtering and prediction theory.
\newblock {\em Journal of basic engineering}, 83(1):95--108, 1961.

\bibitem{wredenhagen1993curvature}
G.F. Wredenhagen and P.R. Belanger.
\newblock Curvature properties of the algebraic {R}iccati equation.
\newblock {\em Systems \& control letters}, 21(4):285--287, 1993.

\bibitem{schechter1959inversion}
S.~Schechter.
\newblock On the inversion of certain matrices.
\newblock {\em Mathematics of Computation}, 13(66):73--77, 1959.

\end{thebibliography}
\end{document}